\documentclass[11pt]{article}
  
\usepackage[bookmarks,colorlinks,breaklinks]{hyperref}  
\hypersetup{linkcolor=blue,citecolor=blue,filecolor=blue,urlcolor=blue} 
\usepackage{fullpage,appendix}
\usepackage{amsmath,amsfonts,amsthm,amssymb,xspace,bm}

\newcommand{\newref}[2][]{\hyperref[#2]{#1~\ref*{#2}}}
\renewcommand{\eqref}[1]{\hyperref[#1]{(\ref*{#1})}}
\numberwithin{equation}{section}

\newcommand{\tref}[1]{\newref[Theorem]{#1}}
\newcommand{\lref}[1]{\newref[Lemma]{#1}}

\newcommand{\cref}[1]{\newref[Corollary]{#1}}

\newcommand{\eref}[1]{\newref[Equation]{#1}}

\theoremstyle{plain}
\newtheorem{theorem}{Theorem}[section]
\newtheorem{lemma}[theorem]{Lemma}

\newtheorem{corollary}[theorem]{Corollary}

\theoremstyle{definition}

\DeclareMathOperator*{\pr}{\mathsf{Pr}} 

\DeclareMathOperator*{\ex}{\mathbb{E}}
\newcommand{\rgta}{\rightarrow}
\newcommand{\lfta}{\leftarrow}
\newcommand{\iprod}[2]{\langle #1,#2\rangle}

\newcommand{\rn}{\mathbb{R}^n}
\newcommand{\rnn}{\mathbb{R}^{n \times n}}

\newcommand{\reals}{\mathbb{R}}

\newcommand{\dpm}{\{1,-1\}}

\newcommand{\note}[1]{\marginpar{\tiny *note in TeX*}}
\newcommand{\ignore}[1]{}

\newcommand{\calD}{{\cal D}}

\renewcommand{\phi}{\varphi}
\renewcommand{\epsilon}{\varepsilon}

\newcommand{\poly}{\mathrm{poly}}

\newcommand{\zo}{\{0,1\}}

\newcommand{\cals}{\mathcal{S}}
\newcommand{\calp}{\mathcal{P}}
\newcommand{\calA}{\mathcal{A}}
\title{Almost Optimal Explicit Johnson-Lindenstrauss Transformations}
\author{Raghu Meka\\
UT Austin\\
{\tt raghu@cs.utexas.edu}}
\date{}       
\begin{document}
\begin{titlepage}
\maketitle 
\thispagestyle{empty}
\begin{abstract}
  The Johnson-Lindenstrauss lemma is a fundamental result in probability with several applications in the design and analysis of algorithms in high dimensional geometry. Most known constructions of linear embeddings that satisfy the Johnson-Lindenstrauss property involve randomness. We address the question of explicitly constructing such embedding families and provide a construction with an almost optimal use of randomness: For $0 < \delta,\epsilon < 1$, we give an explicit generator $G:\zo^r \rgta \reals^{s \times n}$ for $s = O(\log (1/\delta)/\epsilon^2)$ such that for all $w \in \rn$, $\|w\| = 1$,
\[ \pr_{y \in_u \zo^r}[\, |\|G(y)w\|^2 -1| > \epsilon\,] \leq \delta,\]
and seed-length $r = O\left(\log (n/\delta) \cdot \log\left(\frac{\log(n/\delta)}{\epsilon}\right)\right)$. In particular, for $\delta = 1/\poly(n)$ and fixed $\epsilon > 0$ we get seed-length $O(\log n \log \log n)$. Previous constructions required at least $O(\log^2 n)$ random bits to get polynomially small error.
\end{abstract}
\end{titlepage}
\section{Introduction}
The celebrated Johnson-Lindenstrauss lemma (JLL) \cite{JohnsonL84} is by now a standard technique for handling high dimensional data. Among its many known variants (see \cite{ArriagaV99, DasguptaG03, IndykM98, Matousek08}), we use the following version originally due to Achlioptas \cite{Achlioptas2003} as reference \footnote{Throughout, $C$ denotes a universal constant. For a multiset $S$, $x \in_u S$ denotes a uniformly random element of $S$.}.
\begin{theorem}[Achlioptas]\label{th:jllbinary}
For all $w \in \reals^n$, $\|w\| = 1$, $\epsilon > 0$, $s = C\log(1/\delta)/\epsilon^2$,
\[ \pr_{A \in_u \dpm^{s \times n}}[\; |\, \|(1/\sqrt{s})\,A w\|^2 - 1\,| \geq \epsilon\;] \leq \delta.\]
\end{theorem}
We say a family of random matrices has the JL property (or is a JL family) if a similar condition holds. In typical applications of JLL, $\delta$ is taken to be  $1/\poly(n)$ and the goal is to embed a given set of $\poly(n)$ points in $n$ dimensions to $O(\log n)$ dimensions with distortion at most $\epsilon$ for a fixed constant $\epsilon$. This is the setting we concern ourselves with. 

Most results on embedding the Euclidean space as above are probabilistic in nature. However, a simple probabilistic argument shows that there exists a fixed collection of $\poly(n,1/\delta)$ linear mappings satisfying the JL property. Despite much attention, the best known constructions of JL families use at least $O(\log n \log(1/\delta))$ random bits \cite{ClarksonW09}. Besides being a natural problem in geometry as well as derandomization, an explicit family of Johnson-Lindenstrauss transformations would likely help derandomize other geometric algorithms and metric embedding constructions. Further, having an explicit construction is of fundamental importance for streaming algorithms as storing the entire matrix (as opposed to the randomness required to generate the matrix) is often too expensive in the streaming context. 

Our main result is an explicit generator that takes roughly $O(\log n \log\log n)$ random bits and outputs a matrix $A \in \reals^{s \times n}$ satisfying the JL property. 

\begin{theorem}\label{th:main}
  For every $0 < \epsilon,\delta < 1$, there exists an explicit generator $G:\zo^r \rgta \reals^{n \times s}$ for $s = C\log(1/\delta)/\epsilon^2$, such that for every $w \in \rn$, $\|w\| = 1$,
\[ \pr_{y \in_u \zo^r}[\, |\,\|G(y)w\|^2 - 1\,| > \epsilon\,] \leq \delta.\]
The seed-length of the generator is $r = O\left(\log (n/\delta) \cdot \log\left(\frac{\log(n/\delta)}{\epsilon}\right)\right)$.
\end{theorem}

Our construction is elementary in nature using only standard tools in derandomization such as $k$-wise independence and oblivious samplers \cite{Goldreich97}. The construction has the additional property that for $\delta = 1/\poly(n)$, the matrix-vector product $G(y)w$ can be computed efficiently in time $O(n\log n)$. The computational efficiency does not follow directly from the dimensions of $G(y)$, as in our construction $G(y)$ is obtained by composing several matrices some of which are of dimension $\tilde{O}(\sqrt{n}) \times n$. Nevertheless, the large matrices are obtained from the discrete Fourier transform matrix facilitating fast matrix-vector product computations. 

Further, as one of the motivations for derandomizing JLL is its potential applications in streaming, it is important that the entries of the generated matrices be computable in small space. We observe that for any $i \in [s], j \in [n], y \in \zo^r$, the entry $G(y)_{ij}$ can be computed in space $O(\log n)$ and time $O(n^{2+o(1)})$ (for fixed $\epsilon$, $\delta > 1/\poly(n)$). (See proof of \tref{th:main} for the exact bound)

\subsection{Related Work}
Independently, Kane and Nelson \cite{KaneN10} obtained a construction that is similar in spirit to ours and achieves a slightly better seed-length of $r = O(\log n + \log(1/\delta)\log(\log(1/\delta)/\epsilon))$. Note that for the important case of $\delta$ polynomially small, our seed-length is the same as theirs.

The $\ell_2$ streaming sketch of Alon et al.~\cite{AlonMS99} implies an explicit JL family with seed-length $O(\log n)$ for embedding $\reals^n$ into $\reals^s$ with distortion $\epsilon$ and error $\delta$, where $s = O(1/\epsilon^2\delta)$. Karnin et al.~\cite{KarninRS2009} construct an explicit family for embedding $\reals^n$ into $\reals^s$ with distortion $\epsilon = 1/s^{-C}$ and error $\delta = 1/s^{-c}$. The seed-length they achieve is $(1+o(1))\log n + O(\log^2 s)$. 

The works of Diakonikolas et al.~\cite{DKN10} and Meka and Zuckerman \cite{MekaZ10} construct pseudorandom generators for degree $2$ threshold functions achieving a seed-length of $\log n\cdot \poly(1/\delta)$ for fooling with error at most $\delta$. As derandomizing the JL lemma is a special case of fooling degree 2 PTFs, these works give a JL family with seed-length $\log n \cdot \poly(1/\delta)$. 

The best known explicit JL family is the construction of Clarkson and Woodruff \cite{ClarksonW09} who show that a random scaled Bernoulli matrix with $O(\log(1/\delta))$-wise independent entries satisfies the JL lemma. We make use of their result in our construction.

We also note that there are efficient non-black box derandomizations of JLL, \cite{EngebretsenIO02}, \cite{Sivakumar2002}. These works, take as input $N$ points in $\reals^n$, and deterministically compute an embedding (that depends on the input set) into $\reals^{O(\log N)/\epsilon^2}$ which preserves all pairwise distances between the given set of $N$ points.

Finally, we remark that our goal as well as result is very different from those of the recent works \cite{AilonC09, AilonL09, DasguptaKS10, AilonL10, KaneN10} on {\sl fast} or {\sl sparse} Johnson-Lindenstrauss transformations as pioneered by the seminal work of Ailon and Chazelle \cite{AilonC09}. The goal in these works is to design a family of embedding matrices for which the matrix-vector products $Ax$ can be computed efficiently (usually $O(n\log n)$) and are mainly concerned with the setting where the desired error probability $\delta$ is exponentially small. In contrast, we are mainly interested in the case where $\delta$ is polynomially small but want to save on randomness.

\subsection{Outline of Construction}
Our construction is based on a simple iterative scheme: We reduce the dimension from $n$ to $\tilde{O}(\sqrt{n})$ using $k$-wise independence and {\sl oblivious samplers} \cite{Goldreich97} and iterate for $O(\log\log n)$ steps.

Fix a vector $w \in \rn$ with $\|w\| = 1$ and let $\delta = 1/\poly(n)$. We first use an idea of Ailon and Chazelle \cite{AilonC09} who give a family of unitary transformations $\mathcal{R}$ from $\rn \rgta \rn$ such that for every $w \in \rn$ and $V \in_u \mathcal{R}$, the vector $Vw$ is {\sl regular} in the sense that $\|Vw\|_\infty = O(\sqrt{\log n/n})$. We derandomize their construction using limited independence to get a bound of $\|Vw\|_\infty = O(n^{-1/4})$.

We next observe that for a vector $w \in \rn$, with $\|w\|_\infty = O(n^{-1/4} \|w\|_2)$ projecting onto a random set of $O(n^{1/2}\log(1/\delta)/\epsilon^2)$ coordinates preserves the $\ell_2$ norm with distortion at most $\epsilon$. We then note that the random set of coordinates can be chosen using efficient samplers as in \cite{Goldreich97}. The idea of using samplers is due to Karnin et al.~\cite{KarninRS2009} who use samplers for a similar purpose. 

Finally, iterating the above scheme $O(\log \log n)$ times we obtain an embedding of $\rn$ to $\reals^{\poly\log n}$ using roughly $O(\log n\log\log n)$ random bits. We then apply the result of Clarkson and Woodruff and perform the final embedding into $O(\log(1/\delta)/\epsilon^2)$ dimensions by using a random scaled Bernoulli matrix with $O(\log(1/\delta))$-wise independent entries. 
\section{Preliminaries}
Let $H_n \in \{-1/\sqrt{n},1/\sqrt{n}\}^{n \times n}$ be the normalized Hadamard matrix such that $H_n^T H_n = I_n$ (we drop the suffix $n$ when dimension is clear from context). While the Hadamard matrix is known to exist for powers of $2$, for clarity, we ignore this minor technicality and assume that it exists for all $n$.

We make use of Khintchine-Kahane inequalities (cf.~\cite{LedouxT}).
\begin{lemma}[Khintchine-Kahane]\label{lm:kkineq}
  For every $w \in \rn$, $x \in_u \dpm^n$, $k > 0$,
\[ \ex[\,|\iprod{w}{x}|^k\,] \leq k^{k/2} \ex[\,|\iprod{w}{x}|^2\,]^{k/2} = k^{k/2}\|w\|^k.\]
\end{lemma}

We use efficient oblivious samplers as in \cite{Goldreich97}.
\begin{theorem}\label{th:sampler}
For every $\epsilon, \delta$ there exists $s = C\log(1/\delta)/\epsilon^2$ such that the following holds. There exists an explicit collection of subsets of $[n]$, $\mathcal{S}(n,\epsilon,\delta)$, with each $S \in \cal{S}$ of cardinality $|S| = s$, and $|\mathcal{S}| = n  \cdot \poly(1/\epsilon,1/\delta)$ such that for every function $f:[n] \rgta [0,1]$, 
\[ \pr_{S \in_u \cals}\left[\,\left| \frac{1}{s} \sum_{i \in S} f(i) - \ex_{i \in_u [n]}f(i)\right| > \epsilon\,\right] \leq \delta.\]
\end{theorem}

\begin{corollary}\label{cor:sampler}
For every $\epsilon, \delta, B > 0$ there exists $s = C\log(1/\delta)B^2/\epsilon^2$ such that the following holds. There exists an explicit collection of subsets of $[n]$, $\mathcal{S}(n,B,\epsilon,\delta)$, with each $S \in \cal{S}$ of cardinality $|S| = s$, and $|\mathcal{S}| = n  \cdot \poly(B/\epsilon,1/\delta)$ such that for every function $f:[n] \rgta [0,B]$, 
\[ \pr_{S \in_u \cals}\left[\,\left| \frac{1}{s} \sum_{i \in S} f(i) - \ex_{i \in_u [n]}f(i)\right| > \epsilon\,\right] \leq \delta.\]
\end{corollary}
\begin{proof}
Follows by taking $\cals = \cals(n,\epsilon/B,\delta)$ as in \tref{th:sampler} and using the condition of \tref{th:sampler} for $\bar{f}:[n] \rgta [0,1]$ defined by $\bar{f}(i) = f(i)/B$.
\ignore{
  Define $g:[n] \rgta [0,1]$ by $g(i) = f(i)/B$. Let $\cals \equiv \cals(n,\epsilon/B,\delta)$ be as in \tref{th:sampler}. Then, for $S \in_u \cals$,
  \begin{align*}
    \pr\left[\,\left| \frac{1}{s} \sum_{i \in S} f(i) - \ex_{i \in_u [n]}f(i)\right| > \epsilon\,\right] = \pr\left[\,\left| \frac{1}{s} \sum_{i \in S} g(i) - \ex_{i \in_u [n]}g(i)\right| > \epsilon/B\,\right] \leq \delta.
  \end{align*}}
\end{proof}

Finally, we use the following result of Clarkson and Woodruff.
\begin{theorem}[Theorem 2.2, \cite{ClarksonW09}]\label{th:jllbounded}
There exist constants $c,C$ such that the following holds. For $0 < \epsilon, \delta < 1$, $s = c\log(1/\delta)/\epsilon^2$, let $A \in \reals^{s \times n}$ be a random matrix with entries in $\{-1/\sqrt{s}, 1/\sqrt{s}\}$ that are $(C\log(1/\delta))$-wise independent. Then, for every $w \in \rn$, $\|w\| = 1$,
\[ \pr_{A}[\; |\, \|\,A w\|^2 - 1\,| \geq \epsilon\;] \leq \delta.\]
\end{theorem}

Note that for all $k,m$, constructions of $k$-wise independent spaces over $\dpm^m$ with seed-length $O(k\log m)$ are known.

\section{Main Construction}
Suppose that $\delta > 1/n^c$ for a constant $c > 0$. If not we first embed the input vector into $\reals^N$ for $N = \lceil 1/\delta \rceil$ by retaining the first $n$ coordinates and setting the other coordinates to be $0$. The parameters we get by working over $\reals^N$ will be the same as those of \tref{th:recurse} and hence of \tref{th:main}. Further, we assume that $\log(1/\delta)/\epsilon^2 = o(n)$ as else JLL is not interesting.

As outlined in the introduction, we first give a family of rotations to regularize vectors in $\rn$. For a vector $x \in \rn$, let $D(x) \in \rnn$ be the diagonal matrix with $D(x)_{ii} = x_i$. 
\begin{lemma}\label{lm:regularize}
  Let $x \in \dpm^n$ be drawn from a $k$-wise independent distribution. Then, for every $w \in \rn$ with $\|w\| = 1$, 
\[ \pr[\, \|HD(x) w\|_\infty > n^{-(1/2-\alpha)}\, ] \leq \frac{k^{k/2}}{n^{\alpha k-1}}.\] 
\end{lemma}
\begin{proof}
  Let $v = HD(x) w$. Then, for $i \in [n]$, $v_i = \sum_j H_{ij} x_j w_j$ and $\ex[v_i^2] = \sum_j H_{ij}^2 w_j^2 = 1/n$. By Markov's inequality and Khintchine-Kahane inequality \lref{lm:kkineq},
  \begin{align*}
    \pr[\, |v_i| > n^{-(1/2-\alpha)}\, ] &\leq \ex[v_i^k]\cdot n^{(1/2-\alpha)k} \leq k^{k/2} n^{(1/2-\alpha)k}/n^{k/2} = k^{k/2} n^{-\alpha k}.
  \end{align*}
The claim now follows from a union bound over $i \in [n]$. 
\end{proof}

We now give a family of transformations for reducing $n$ dimensions to $\tilde{O}(n^{1/2})$ dimensions with distortion at most $\epsilon$. For $S \subseteq [n]$, let $\calp_S:\rn \rgta \reals^{|S|}$ be the projection onto the coordinates in $S$. 

\begin{lemma}\label{lm:dimreduce}
Let $\cals \equiv \cals(n,n^{1/4},\epsilon,\delta)$, $s = O(n^{1/2} \log(1/\delta)/\epsilon^2)$ be as in \cref{cor:sampler} and let $\calD$ be a $k$-wise independent distribution over $\dpm^n$. For $S \in_u \cals$, $x \lfta D$, let random linear transformation $A_{S,x}:\rn \rgta \reals^s$ be defined by $A_{S,x} = \sqrt{n/s}\cdot \calp_S \cdot H D(x)$. Then, for every $w \in \rn$ with $\|w\| = 1$,
\[ \pr[\, |\|A_{S,x}(w)\|^2 - 1 | \geq \epsilon\,] \leq \delta + k^{k/2}/n^{k/8-1}.\]
\end{lemma}
\begin{proof}
  Let $v = HD(x) w$. Then, $\|v\| = 1$ and by \lref{lm:regularize} applied for $\alpha = 1/8$,
\[ \pr[\, \|v\|_\infty > n^{-3/8}\,] \leq k^{k/2}/n^{k/8-1}.\]
Now suppose that $\|v\|_\infty \leq n^{-3/8}$. Define $f:[n] \rgta \reals$ by $f(i) = n \cdot v_i^2 \leq n \cdot n^{-3/4} = n^{1/4} = B$. Then,
\[ \|A_{S,x}(w)\|^2 = (n/s) \| \calp_S(v)\|^2 = \frac{1}{s} \sum_{i \in S} n v_i^2 = \frac{1}{s} \sum_{i \in S} f(i),\]
and $\ex_{i \in_u [n]} f(i) = (1/n) \sum_i n \cdot v_i^2 = 1$. Therefore, by \cref{cor:sampler},
  \begin{align*}
   \pr[\, |\,\|A_{S,x}(w)\|^2 - 1\,| \geq \epsilon\,] &= \pr_S\left[\,\left|\frac{1}{s} \sum_{i \in S} f(i) - \ex_{i \in_u [n]} f(i)\right|\geq \epsilon\,\right] \leq \delta.
  \end{align*}
The claim now follows.
\end{proof}

We now recursively apply the above lemma. Fix $\epsilon,\delta > 0$. Let $\calA(n,k):\rn \rgta \reals^{s(n)}$ be the collection of transformations $\{A_{S,x}: S \in_u \cals, x \lfta D\}$ as in the above lemma for $s(n) =  s(n,n^{1/4},\epsilon,\delta) = Cn^{1/2} \log(1/\delta)/\epsilon^2$. Note that we can sample from $\calA(n,k)$ using $O(k \log n + \log(1/\epsilon) + \log(1/\delta))$ random bits. Let $n_0 = n$, and let $n_{i+1} = s(n_i)$. Let $k_0 = 16(c+1)$ (recall that $\delta > 1/n^c$) and $k_{i+1} = 2^i k_0$. The parameters $n_i,k_i$ are chosen so that $1/n_i^{k_i}$ is always polynomially small. Fix $t > 0$ to be chosen later so that $k_i < n_i^{1/8}$ for $i < t$.

\begin{lemma}\label{lm:recursion}
  For $A_0 \in_u \calA(n_0,k_0), A_1 \in_u \calA(n_1,k_1), \cdots, A_{t-1} \in_u \calA(n_{t-1},k_{t-1})$ chosen independently, and $w \in \rn$, $\|w\| = 1$,
\[ \pr[\, (1-\epsilon)^t \leq \|A_{t-1} \cdots A_1 A_0 (w)\|^2 \leq (1+\epsilon)^t\,] \geq 1 - t \delta - \sum_{i=0}^{t-1} \frac{k_i^{k_i/2}}{n_i^{k_i/8-1}}.\]
\end{lemma}
\begin{proof}
  The proof is by induction on $i = 1,\ldots,t$. For $i=1$, the claim is same as \lref{lm:dimreduce}. Suppose the statement is true for $i-1$ and let $v = A_{i-1} \cdots A_0 (w)$. Then, $v \in \reals^{n_{i}}$ and by \lref{lm:dimreduce} applied to $\calA(n_{i},k_{i})$, 
\[ \pr[\, (1-\epsilon) \|v\|^2 \leq \|A_{i}(v)\|^2 \leq (1+\epsilon)\|v\|^2\,] \geq 1- \delta -  \frac{k_i^{k_i/2}}{n_i^{k_i/8-1}}.\]
The claim now follows by induction.
\end{proof}
What follows is a series of simple calculations to bound the seed-length and error from the above lemma. Observe that 
\begin{equation}\label{eq:rec1}
 n^{(1/2)^i} \leq n_i = n^{(1/2)^i} \cdot \left(\frac{C\log(1/\delta)}{\epsilon^2}\right)^{1 + (1/2) + \cdots + (1/2)^{i-1}} \leq n^{(1/2)^i} \left(\frac{C\log(1/\delta)}{\epsilon^2}\right)^2.  
\end{equation}

Let $t = O(\log \log n)$ be such that $2^t = \log n/8\log \log n$. Then, $n_t \leq \log^8 n \cdot (C\log(1/\delta)/\epsilon^2)^2$, and for $i < t$, 
\begin{equation}\label{eq:rec2}
k_i < k_t = 16(c+1) 2^t = 2(c+1)\log n/\log \log n < \log n = n^{(1/2)^t/8} < n_t^{1/8} < n_i^{1/8},  
\end{equation}
where we assumed that $\log\log n > 2c+2$.
Therefore, the error in \lref{lm:recursion} can be bounded by
\begin{align*}
t \delta + \sum_{i=0}^{t-1} \frac{k_i^{k_i/2}}{n_i^{k_i/8-1}} &\leq t \delta + n \sum_{i=0}^{t-1} n_i^{-k_i/16} &\text{\hspace{.3in} (\eref{eq:rec2})}\\
&\leq t \delta + n \sum_{i=0}^{t-1}   n^{- (1/2)^i \cdot 16(c+1) \cdot 2^i/16}  &\text{\hspace{.3in} (\eref{eq:rec1})}\\
&\leq t \delta + t/n^c \leq 2 t \delta  &\text{\hspace{.3in} (as $\delta > 1/n^c$)}.
\end{align*}
Note that,
\begin{multline*}
 k_i \log n_i \leq 16(c+1)\cdot 2^i (\log n/2^i + 2C\log\log(1/\delta) + 4\log(1/\epsilon)) = 
O(\log n + \log n \log(1/\epsilon)/\log\log n).
\end{multline*}

Therefore, the randomness needed after $t = O(\log \log n)$ iterations is 
\begin{align*}
\sum_{i=0}^{t-1} O(k_i \log n_i + \log(1/\epsilon) + \log(1/\delta)) = O(\log n \cdot \log(\log n/\epsilon)).
\end{align*}
Combining the above arguments (and simplifying the resulting expression for seed-length) we get:
\begin{theorem}\label{th:recurse}
There exists an explicit generator that takes $O(\log(n/\delta)\cdot \log(\,\log(n/\delta)/\epsilon\,))$ random bits and outputs a linear transformation $A:\rn \rgta \reals^m$ for $m = O((\log^{10}(n/\delta))/\epsilon^4)$, so that for every $w \in \rn$, $\|w\| = 1$,
\[ \pr[ \,|\;\|Aw\|^2 - 1\;|> (C\log\log n)\epsilon\,] \leq (C\log\log n)\delta.\]
\end{theorem}
We can now obtain our final construction of explicit Johnson-Lindenstrauss families by composing the above family with that of \tref{th:jllbounded}.

\begin{proof}[Proof of \tref{th:main}]
  Follows by composing the transformations of the above theorem for $\epsilon' = \epsilon/C\log\log n$, $\delta' = \delta/C\log\log n$ with those of \tref{th:jllbounded} using $O(\log(1/\delta))$-wise independence. The additional randomness required is $O(\log(1/\delta)\log m)= O(\log (1/\delta)(\log\log (n/\delta) + \log(1/\epsilon))$. 

We next bound the time for computing matrix-vector products for the matrices we output. Let $\delta > 1/n^c$. Note that for $i < t$, the matrices $A_i$ of \lref{lm:recursion} are of the form $\calp_S \cdot H_{n_i} D(x)$ for a $k$-wise independent string $x \in \dpm^{n_i}$. Thus, for any vector $w_i \in \reals^{n_i}$, $A_i w_i$ can be computed in time $O(n_i \log n_i)$ using the discrete Fourier transform. Therefore, for any $w = w_0 \in \reals^{n_0}$, the product $A_{t-1} \cdots A_1 A_0 w_0$ can be computed in time 
\begin{align*}
\sum_{i=0}^{t-1} O(n_i \log n_i) &\leq O(n \log n) + \log n \cdot \sum_{i=1}^{t-1} O\left(n^{1/2^i} (\log(1/\delta)/\epsilon^2)^{2}\right)&\text{ (\eref{eq:rec1})}\\
&= O(n \log n  + \sqrt{n} \log n\log^2(1/\delta)/\epsilon^4).
\end{align*}
It is easy to check that the above bound dominates the time required to perform the final embedding.

A similar calculation shows that for indices $i \in s, j \in [n]$, the entry $G(y)_{ij}$ of the generated matrix can be computed in space $O\left(\sum_i \log n_i\right) = O(\log n + \log(1/\epsilon)\cdot \log\log n)$ by expanding the product of matrices and enumerating over all intermediary indices. The time required to perform the calculation is $O(s \cdot n_t \cdot n_{t-1} \cdots n_0) = n^2\cdot (\log n/\epsilon)^{O(\log \log n)}$. 
\end{proof} 
\newpage
\bibliographystyle{alpha}
\bibliography{refcyclic}
\end{document}